\documentclass[journal, draftclsnofoot,oneside, onecolumn,11pt, letterpaper]{IEEEtran}

\ifCLASSINFOpdf

\else

\fi

\usepackage{pslatex}
\usepackage{amsfonts,color,morefloats}
\usepackage{amssymb,amsmath,latexsym,amsthm}
\usepackage{bbm}
\usepackage{amsfonts}
\usepackage{mathrsfs}
\usepackage{amsthm}
\usepackage{latexsym}
\usepackage{color}
\usepackage{dcolumn}
\usepackage{multicol}
\usepackage{extarrows}
\usepackage{amsmath}
\usepackage{cite}
\usepackage{amsmath,bm}
	\usepackage{cite}
\usepackage[numbers,sort&compress]{natbib}
	\usepackage{color}
	\allowdisplaybreaks[4]

\newtheorem{theorem}{Theorem}
\newtheorem{lemma}{Lemma}

\newtheorem{definition}{Definition}
\newtheorem{example}{Example}

\newtheorem*{remark}{Remark}

\makeatletter

\newcommand{\Rmnum}[1]{\expandafter\@slowromancap\romannumeral #1@}
\makeatother

\begin{document}

	\title{Arithmetic crosscorrelation of binary $\bm{m}$-sequences with coprime periods
	}

	\author{Xiaoyan~Jing,
			Keqin~Feng
			\thanks{Xiaoyan Jing is with Research Center for Number Theory and Its Applications, Northwest University, Xi'an 710127, China (e-mail: jxymg@126.com).}
            \thanks{Feqin Feng is with Department of Mathematical Sciences, Tsinghua University, Beijing 100084, China (e-mail: fengkq@tsinghua.edu.cn).}
			
		}

\markboth{}%
{Shell \MakeLowercase{\textit{et al.}}: Bare Demo of IEEEtran.cls for IEEE Journals}

\maketitle

\IEEEpeerreviewmaketitle

\begin{abstract}
The arithmetic crosscorrelation of binary $\bm{m}$-sequences with coprime periods $2^{n_1}-1$ and $2^{n_2}-1$\ ($\gcd(n_1,n_2)=1$) is determined. The result shows that the absolute value of arithmetic crosscorrelation of such  binary $\bm{m}$-sequences is not greater than $2^{\min(n_1,n_2)}-1$.

\end{abstract}

\begin{IEEEkeywords}
arithmetic crosscorrelation, binary $\bm{m}$-sequences, trace
mapping
\end{IEEEkeywords}

\IEEEpeerreviewmaketitle

\section{Introduction and Main Result}\label{sec1}
Let $n\geq2$, $q=2^n$ and $\mathbb{F}_q$ be the finite field with $q$ elements, and
$$T: \mathbb{F}_{q}\rightarrow\mathbb{F}_2, T(a)= a+ a^2+ a^{2^2}+\cdots+ a^{2^{{n}-1}}\quad ( a\in\mathbb{F}_{q})$$ be the trace
mapping from $\mathbb{F}_{q}$ to $\mathbb{F}_{2}$. Let $\alpha$ be a primitive element of $\mathbb{F}_{q}^{*}$ which means $\mathbb{F}_{q}^{*}=\mathbb{F}_{q}\backslash\{0\}=\langle\alpha\rangle=\{1=\alpha^0, \alpha^1, \ldots, \alpha^{q-2}\}\ (\alpha^{q-1}=1)$. For each $c\in\mathbb{F}_{q}^{*}$, the binary sequence
$$\bm{A}_c=\{T(c\alpha^{\lambda})\}^{q-2}_{\lambda=0}$$
is a periodical sequence with period $q-1$, called binary $\bm{m}$-sequence. All sequences $\bm{A}_c(c\in\mathbb{F}_{q}^{*})$ are (cyclically) shifted equivalent to each other, which means, for $c=\alpha^{\tau}$, $\bm{A}_c=\bm{A}_1^{(\tau)}$ where $\bm{A}_1=\{a_{\lambda}=T(\alpha^{\lambda})\}^{q-2}_{\lambda=0}$ and $\bm{A}_1^{(\tau)}=\{b_{\lambda}\}^{q-2}_{\lambda=0}$ is a shifted sequence of $\bm{A}_1$ by $b_{\lambda}=a_{\lambda+\tau}$.

Due to their good pseudo-randomness and ideal autocorrelation, $\bm{m}$-sequences have significant applications
 in information science and engineering, as key sequence design, data encryption, spread spectrum communication and synchronous communication.

 In this paper, we concern on the arithmetic crosscorrelation of binary $\textbf{\emph{m}}$-sequences. Arithmetic correlation is another description of sequence correlation different from the classical correlation, proposed by  D. Mandelbaum in 1967
(\cite{9Mand}) in his research on arithmetic codes. Later, this notion was generalized to the nonbinary case and developed by Goresky and Klapper (\cite{3G1,4G2,5G3,6G4}) in relation to their research on sequences generated by the linear shift registers with carry.

 Now we state the definition of arithmetic autocorrelation and  crosscorrelation for (periodic) binary sequences. Let
 $$ {\bm{A}}=(a_{\lambda})^{N-1}_{\lambda=0},\  {\bm{B}}=(b_{\lambda})^{N-1}_{\lambda=0}\ (a_{\lambda}, b_{\lambda}\in\{0,1 \})$$
be two binary sequences with same period $N$. The classical crosscorrelation of $\bm{A}$ and $\bm{B}$ is defined by
$$\mathcal{C}(\bm{A}, \bm{B})=\sum^{N-1}_{\lambda=0}(-1)^{a_{\lambda}+b_{\lambda}}\in\mathbb{Z}.$$
For $0\leq\tau\leq N-1$, let $\bm{B}^{(\tau)}=\{b_{\lambda}'=b_{\lambda+\tau}\}^{N-1}_{\lambda=0}$ be the $\tau$-shifted sequence of $\bm{B}$ and
$$\mathcal{C}(\bm{A}, \bm{B}, \tau)=\mathcal{C}(\bm{A}, \bm{B}^{(\tau)})=\sum^{N-1}_{\lambda=0}(-1)^{a_{\lambda}+b_{\lambda+\tau}}\in\mathbb{Z}.$$
Then $\mathcal{C}(\bm{A}, \bm{B})=\mathcal{C}(\bm{A}, \bm{B}, 0)$. If $\bm{B}=\bm{A}$, then $\mathcal{C}(\bm{A}, \bm{A})=N$ and $\mathcal{C}^{(\tau)}(\bm{A})=\mathcal{C}(\bm{A}, \bm{A}^{(\tau)})\ (1\leq\tau\leq N-1)$ are called the classical autocorrelation of $\bm{A}$. It is known that the binary $\bm{m}$-sequences $\bm{A}$ with period $2^n-1$ have ideal autocorrelation $\mathcal{C}^{(\tau)}(\bm{A})=-1$ for all $\tau$, $1\leq\tau\leq 2^n-2$.

\begin{definition}
  Let $$I({\bm{A}})=\sum_{\lambda=0}^{N-1}a_{\lambda}2^{\lambda},\ I({\bm{B}})=\sum_{\lambda=0}^{N-1}b_{\lambda}2^{\lambda}\in\mathbb{Z}.$$
The integer $|I({\bm{A}})-I({\bm{B}})|\geq0$ can also be expressed in 2-adic expansion $$|I({\bm{A}})-I({\bm{B}})|=\sum_{\lambda=0}^{N-1}c_{\lambda}2^{\lambda} \ (c_{\lambda}\in\{0, 1\}).$$
Let $I_0=\sharp\{0\leq\lambda\leq N-1|c_{\lambda}=0\}$, $I_1=\sharp\{0\leq\lambda\leq N-1|c_{\lambda}=1\}$.

We define
\begin{align*}
		\mathcal{M}(\bm{A},\bm{B})=\left\{\begin{array}{ll}
I_0-I_1,\ \text{if}\ I({\bm{A}})\geq I({\bm{B}}),\\
I_1-I_0,\ \text{if}\ I({\bm{A}})<I({\bm{B}}).
\end{array}\right.
\end{align*}
The arithmetic crosscorrelation of ${\bm{A}}$ and ${\bm{B}}$ is the values
$$\mathcal{M}^{(\tau)}(\bm{A},\bm{B})=\mathcal{M}(\bm{A},\bm{B}^{(\tau)})\ (0\leq\tau\leq N-1).$$
If $\bm{B}=\bm{A}$ and $N$ is the minimal period of $\bm{A}$, then $\mathcal{M}(\bm{A},\bm{A})=\mathcal{M}^{(0)}(\bm{A},\bm{A})=N$ and the arithmetic autocorrelation of $\bm{A}$ is the values
 $$\mathcal{M}^{(\tau)}(\bm{A})=\mathcal{M}^{(\tau)}(\bm{A},\bm{A})=\mathcal{M}(\bm{A},\bm{A}^{(\tau)})\ (1\leq\tau\leq N-1).$$
\end{definition}
Comparing with the classical correlation, there are currently few results on the arithmetic autocorrelation and arithmetic crosscorrelation of binary sequences. Large families of binary sequences with ideal arithmetic autocorrelation $(\mathcal{M}^{(\tau)}(\bm{A})=0$ for $1\leq\tau\leq N-1)$ and ideal arithmetic crosscorrelation $(\mathcal{M}^{(\tau)}(\bm{A},\bm{B})=0$ for $0\leq\tau\leq N-1)$ have been constructed in \cite{9Mand} and \cite{3G1} respectively. For the binary $\bm{m}$-sequences $\bm{A}$ with period $2^n-1$, Chen et. al (\cite{1C1}) raised a conjecture on the absolute values $|\mathcal{M}^{(\tau)}(\bm{A})|\ (1\leq\tau\leq 2^n-2)$ and their multiplicity. This conjecture has been verified in \cite{11J1} by determining the distribution of $\{\mathcal{M}^{(\tau)}(\bm{A})|1\leq\tau\leq 2^n-2\}$. For two binary sequences $\bm{A}$ and $\bm{B}$ with period $N_1$ and $N_2$ respectively, and $\gcd(N_1, N_2)=1$. We view both of $\bm{A}$ and $\bm{B}$ as sequences with same period $N_1N_2$. Then Chen et. al (\cite{2C2}) proved that all $\mathcal{M}^{(\tau)}(\bm{A},\bm{B})\ (0\leq\tau\leq N_1N_2-1)$ have the same value $\mathcal{M}(\bm{A},\bm{B})(=\mathcal{M}^{(0)}(\bm{A},\bm{B}))$.
Particularly for  two binary $\bm{m}$-sequences $\bm{A}$ and $\bm{B}$ with period $N_1=2^{n_1}-1$ and $N_2=2^{n_2}-1$, and $\gcd(n_1, n_2)=1$, then $\gcd(N_1, N_2)=1$,  an upper bounds $|\mathcal{M}(\bm{A},\bm{B})|\leq cn_1\cdot 2^{n_2}$ is present in \cite{1C1} where $n_2>n_1$ and $c>0$ is an absolute constant.

In this paper the exact value $\mathcal{M}(\bm{A},\bm{B})$ is determined as shown in the following Theorem 1. As a direct consequence we get $|\mathcal{M}(\bm{A},\bm{B})|\leq 2^{\min(n_1,n_2)}-1$.

Now we state the main result of this paper.

Let $n_1,\ n_2\geq2$, $\gcd(n_1,n_2)=1$, $q_i=2^{n_i}$ and  $N_i=q_i-1\ (i=1,2)$. Let $\alpha_i$ be a primitive element of the finite field $\mathbb{F}_{q_i}$,
$T_i: \mathbb{F}_{q_i}\rightarrow\mathbb{F}_2$ be the trace mapping, 
$g_i(x)\in\mathbb{F}_2[x]$ be the primitive polynomial and $g_i(\alpha_i)=0\ (i=1,2)$.

\begin{theorem} Under above assumption, let $F(x)=f_0+f_1x+\cdots+f_{l-1}x^{l-1}+x^l \in \mathbb{F}_2[x]\ (0\leq l\leq n_1-1)$ be the unique polynomial satisfying $F(x)g_2(x)\equiv 1\pmod{g_1(x)}$.
Then the arithmetic crosscorrelation  of  binary $\bm{m}$-sequence
${\bm{A}}=(a_{i}=T_1(\alpha_1^{i}))^{N_1N_2-1}_{i=0}$,\  ${\bm{B}}=(b_{i}=T_2(\alpha_2^{i}))^{N_1N_2-1}_{i=0}$
is $$\mathcal{M}({\bm{A}},{\bm{B}})=(-1)^{f_0+1}(2^{n_1-l}-1).$$
Particularly, $|\mathcal{M}({\bm{A}},{\bm{B}})|\leq 2^{\min(n_1,n_2)}-1=\min\{2^{n_1}-1, 2^{n_2}-1\}$, and the equality holds if and only if, for $n_1<n_2$, $g_2(x)\equiv 1\pmod {g_1(x)}$.
\end{theorem}
We prove Theorem 1 in next section.
\begin{remark}
For $n_1=1$, $g_1(x)=1+x$, the polynomial $F(x)$ satisfying $F(x)g_2(x)\equiv 1\pmod{1+x}$ and $\deg F(x)\leq 1-1=0$ is $F(x)=1$. In
this case $f_0=1$, $l=0$ and $(2^{n_1-l}-1)(-1)^{f_0+1}=1$. On the other hand, both of ${\bm{A}}=(a_{i}=1)_{i\geq0}$ and ${\bm{B}}=(b_{i}=T_2(\alpha_2^{i}))^{2^{n_2}-2}_{i=0}$ are considered as binary sequences with period $q_2=2^{n_2}-1$. We have $I(A)=\sum_{i=0}^{q_2-1}2^i$ and $I(B)=\sum_{i=0}^{q_2-1}b_i2^i$. Therefore $I(A)-I(B)=\sum_{i=0}^{q_2-1}(1-b_i)\cdot2^i>0$, and by definition,
$$\mathcal{M}(\bm{A}, \bm{B})=\sharp\{0\leq i\leq q_2-1|b_i=1\}-\sharp\{0\leq i\leq q_2-1|b_i=0\}=2^{n_2-1}-(2^{n_2-1}-1)=1$$
which means that Theorem 1 also holds for the case $n_1=1$ or $n_2=1$.
\end{remark}

\begin{example}
Let $n_1=3$, $n_2=4$, $g_1(x)=x^3+x+1$, $g_2(x)=x^4+x+1$, then we have $N_1N_2=105$, the
 two binary $\bm{m}$-sequences
${\bm{A}}$ and ${\bm{B}}$ are
\begin{align*}
{\bm{A}}=(&1,0,1,1,1,0,0,1,0,1,1,1,0,0,1,0,1,1,1,0,0,1,0,1,1,1,0,0,1,0,1,1,1,0,0,\\
          &1,0,1,1,1,0,0,1,0,1,1,1,0,0,1,0,1,1,1,0,0,1,0,1,1,1,0,0,1,0,1,1,1,0,0,\\
          &1,0,1,1,1,0,0,1,0,1,1,1,0,0,1,0,1,1,1,0,0,1,0,1,1,1,0,0,1,0,1,1,1,0,0),\\
{\bm{B}}=(&1,0,0,1,1,0,1,0,1,1,1,1,0,0,0,1,0,0,1,1,0,1,0,1,1,1,1,0,0,0,1,0,0,1,1,\\
          &0,1,0,1,1,1,1,0,0,0,1,0,0,1,1,0,1,0,1,1,1,1,0,0,0,1,0,0,1,1,0,1,0,1,1,\\
          &1,1,0,0,0,1,0,0,1,1,0,1,0,1,1,1,1,0,0,0,1,0,0,1,1,0,1,0,1,1,1,1,0,0,0).
\end{align*}
By calculating we get $F(x)=x^2+x$, which implies $l=2$ and $f_0=0$, then we have
$$\mathcal{M}({\bm{A}},{\bm{B}})=(-1)^{f_0+1}(2^{n_1-l}-1)=-1.$$
This is consistent with the result of direct calculations by the definition of arithmetic crosscorrelation.
\end{example}

\begin{example}
Let $n_1=3$, $n_2=5$, $g_1(x)=x^3+x^2+1$, $g_2(x)=x^5+x^3+1$, then we have $N_1N_2=217$, the
 two binary $\bm{m}$-sequences
${\bm{A}}$ and ${\bm{B}}$ are
\begin{align*}
{\bm{A}}=(&1,1,1,0,1,0,0,1,1,1,0,1,0,0,1,1,1,0,1,0,0,1,1,1,0,1,0,0,1,1,1,\\
          &0,1,0,0,1,1,1,0,1,0,0,1,1,1,0,1,0,0,1,1,1,0,1,0,0,1,1,1,0,1,0,\\
          &0,1,1,1,0,1,0,0,1,1,1,0,1,0,0,1,1,1,0,1,0,0,1,1,1,0,1,0,0,1,1,\\
          &1,0,1,0,0,1,1,1,0,1,0,0,1,1,1,0,1,0,0,1,1,1,0,1,0,0,1,1,1,0,1,\\
          &0,0,1,1,1,0,1,0,0,1,1,1,0,1,0,0,1,1,1,0,1,0,0,1,1,1,0,1,0,0,1,\\
          &1,1,0,1,0,0,1,1,1,0,1,0,0,1,1,1,0,1,0,0,1,1,1,0,1,0,0,1,1,1,0,\\
          &1,0,0,1,1,1,0,1,0,0,1,1,1,0,1,0,0,1,1,1,0,1,0,0,1,1,1,0,1,0,0),\\
{\bm{B}}=(&1,0,1,0,1,1,1,0,1,1,0,0,0,1,1,1,1,1,0,0,1,1,0,1,0,0,1,0,0,0,0,\\
          &1,0,1,0,1,1,1,0,1,1,0,0,0,1,1,1,1,1,0,0,1,1,0,1,0,0,1,0,0,0,0,\\
          &1,0,1,0,1,1,1,0,1,1,0,0,0,1,1,1,1,1,0,0,1,1,0,1,0,0,1,0,0,0,0,\\
          &1,0,1,0,1,1,1,0,1,1,0,0,0,1,1,1,1,1,0,0,1,1,0,1,0,0,1,0,0,0,0,\\
          &1,0,1,0,1,1,1,0,1,1,0,0,0,1,1,1,1,1,0,0,1,1,0,1,0,0,1,0,0,0,0,\\
          &1,0,1,0,1,1,1,0,1,1,0,0,0,1,1,1,1,1,0,0,1,1,0,1,0,0,1,0,0,0,0,\\
          &1,0,1,0,1,1,1,0,1,1,0,0,0,1,1,1,1,1,0,0,1,1,0,1,0,0,1,0,0,0,0).
\end{align*}
By calculating we get $F(x)=x^2+1$, which implies $l=2$ and $f_0=1$, then we have
$$\mathcal{M}({\bm{A}},{\bm{B}})=(-1)^{f_0+1}(2^{n_1-l}-1)=1.$$
This is consistent with the result of direct calculations by the definition of arithmetic crosscorrelation.
\end{example}

\section{Proof of Theorem 1}\label{sec2}
In this section we fix the following notations.

$n_1,\ n_2\geq2$, and for $i=1, 2$, $q_i=2^{n_i}$, $N_i=q_i-1$, $\mathbb{F}^{\ast}_{q_i}= \langle\alpha_i\rangle$, $\gcd(n_1,n_2)=1$;
$g_i(x)\in\mathbb{F}_2[x]$ is the primitive polynomial, $\deg g_i(x)=n_i$, $g_i(\alpha_i)=0$;
$T_i: \mathbb{F}_{q_i}\rightarrow\mathbb{F}_2$ is the trace mapping from $\mathbb{F}_{q_i}$ to $\mathbb{F}_2$;
${\bm{A}}=(a_{i}=T_1(\alpha_1^{i}))^{N_1-1}_{i=0}$ and ${\bm{B}}=(b_{i}=T_2(\alpha_2^{i}))^{N_2-1}_{i=0}$ are binary $\bm{m}$-sequences with minimal period $N_1$ and $N_2$ respectively. We view both of ${\bm{A}}$ and ${\bm{B}}$ as sequences with same period $N_1N_2$.

Firstly we state some basic facts on binary $\bm{m}$-sequences.
\begin{lemma}\label{lem1}
(1). Let $\bm{S}=(s_i)_{i=0}^{2^n-2}$ be a binary $\bm{m}$-sequence with period $2^n-1\ (n\geq2)$. For $t\geq1$, $\bm{c}=(c_1, \ldots, c_t)\in\mathbb{F}_2^t$, let
$$\mathcal{N}(\bm{S}, \bm{c})=\sharp\{0\leq\lambda\leq 2^n-2|(s_{\lambda}, s_{\lambda+1}, \ldots, s_{\lambda+t-1})=\bm{c}\}.$$
Then for $1\leq t\leq n$
\begin{align*}
\mathcal{N}(\bm{S}, \bm{c})=\left\{
\begin{array}{ll}
2^{n-t},\ &if(c_1, \ldots, c_t)\neq(0, \dots, 0),\\
2^{n-t}-1,\ &if(c_1, \ldots, c_t)=(0, \dots, 0),
\end{array}\right.
\end{align*}
and for $t\geq n+1$ and $\bm{c}=(0, \ldots, 0)\in\mathbb{F}_2^t$, $\mathcal{N}(\bm{S}, \bm{c})=0$.

(2). The binary $\bm{m}$-sequences  ${\bm{A}}$ and ${\bm{B}}$ can be generated by linear shift registers with level $n_1$ and $n_2$ respectively. Then the sequence ${\bm{A}}+{\bm{B}}=\{a_i+b_i\in\mathbb{F}_2\}_{i=0}^{N_1N_2-1}$ can be generated by a linear shift register with level $n_1+n_2$. As a direct consequence, there is no $\lambda\ (0\leq\lambda\leq N_1N_2-1)$ such that $a_i+b_i=0\in\mathbb{F}_2$ for all $i=\lambda, \lambda+1, \ldots, \lambda+n_1+n_2-1$.
\end{lemma}

Next, we introduce a general formula on $\mathcal{M}(\bm{A}, \bm{B})$ presented in \cite{11J1}. Let
$$\bm{C}=\bm{A}+\bm{B}=\{c_i\}_{i=0}^{N_1N_2-1},\ c_i=a_i+b_i=T_1(\alpha_1^i)+T_2(\alpha_2^i)\in\mathbb{F}_2=\{0, 1\}.$$
For $t\geq0$, let
$$N(0,1;t)=\#\{0\leq i\leq N_1N_2-1|
	a_i=0,\ b_i=1, \ a_{i+\lambda}=b_{i+\lambda}\ (1\leq \lambda\leq t),\ a_{i+t+1}\neq b_{i+t+1}\},$$
$$N(1,0;t)=\#\{0\leq i\leq N_1N_2-1|
	a_i=1,\ b_i=0, \ a_{i+\lambda}=b_{i+\lambda}\ (1\leq \lambda\leq t),\ a_{i+t+1}\neq b_{i+t+1}\}.$$
The condition $a_{i+\lambda}=b_{i+\lambda}\ (1\leq \lambda\leq t)$ means that $c_{i+\lambda}=0\in\mathbb{F}_2\ (1\leq \lambda\leq t)$. By Lemma \ref{lem1} (2) we know that for $t\geq n_1+n_2$, $N(0,1;t)=N(1,0;t)=0$.
\begin{lemma}\label{lem2}(\cite{11J1}, Theorem 4)
$\mathcal{M}(\bm{A}, \bm{B})=N_1N_2-2g(\bm{A}, \bm{B})$ where $g(\bm{A}, \bm{B})=\sum\limits_{t=0}^{n_1+n_2-1}(tN(0,1;t)+N(1,0;t)).$
\end{lemma}
In order to reduce this formula on $g(\bm{A}, \bm{B})$, we need:
\begin{lemma}\label{lem3}
$\sum\limits_{t=0}^{n_1+n_2-1}N(1,0;t)=\frac{1}{4}(N_2-1)(N_1+1).$
\end{lemma}
\begin{proof}
The left-hand side is
\begin{align}\label{1}
\sum\limits_{t=0}^{n_1+n_2-1}N(0,1;t)=\#\{0\leq i\leq N_1N_2-1|a_i=T_1(\alpha_1^i)=1,\ b_i=T_2(\alpha_2^i)=0\}
\end{align}
From Lemma \ref{lem1} (1) we know that
\begin{align}\label{2}
\#\{0\leq \lambda\leq N_1-1|T_1(\alpha_1^{\lambda})=1\}=2^{n_1-1},\ \#\{0\leq \mu\leq N_2-1|T_2(\alpha_2^{\mu})=0\}=2^{n_2-1}-1.
\end{align}
By $\gcd(N_1, N_2)=1$ and Chinese Remainder Theorem we have the isomorphism of rings:
$$\varphi: \mathbb{Z}_{N_1}\oplus\mathbb{Z}_{N_2}\simeq\mathbb{Z}_{N_1N_2},\ (\lambda,\ \mu)\mapsto\lambda r_2N_2+\mu r_1N_1\pmod{N_1N_2},$$
where $r_1$ is the integer satisfying $r_1N_1\equiv1\pmod {N_2}$ and $0\leq r_1\leq N_2-1$, $r_2$ is the integer satisfying $r_2N_2\equiv1\pmod {N_1}$ and $0\leq r_2\leq N_1-1$. For each $l\in\mathbb{Z}_{N_1N_2}$, $\varphi^{-1}(l)=(l\pmod{N_1}, l\pmod{N_2})$. Therefore $T_1(\alpha_1^{\lambda})=1$ and $T_2(\alpha_2^{\mu})=0$ if and only if $T_1(\alpha_1^{i})=1$ and $T_2(\alpha_2^{i})=0$ for $i=\varphi((\lambda,\ \mu))=\lambda r_2N_2+\mu r_1N_1$.

Then from (\ref{1}) and (\ref{2}) we get
$$\sum\limits_{t=0}^{n_1+n_2-1}N(1,0;t)=2^{n_1-1}(2^{n_2-1}-1)=\frac{1}{4}(N_2-1)(N_1+1).$$
\end{proof}

\begin{lemma}\label{lem4}
$\mathcal{M}(\bm{A}, \bm{B})=\frac{1}{2}(N_1N_2+N_1-N_2+1)-2\sum\limits_{t=0}^{n_1+n_2-1}tN(0,1;t)$.
\end{lemma}
\begin{proof}
From Lemma \ref{lem2} and \ref{lem3} we get
$$g(\bm{A}, \bm{B})=\sum\limits_{t=0}^{n_1+n_2-1}tN(0,1;t)+\frac{1}{4}(N_2-1)(N_1+1).$$
Therefore $$\mathcal{M}(\bm{A}, \bm{B})=N_1N_2-2g(\bm{A}, \bm{B})=N_1N_2-\frac{1}{2}(N_2-1)(N_1+1)-2\sum\limits_{t=0}^{n_1+n_2-1}tN(0,1;t)$$
and$$N_1N_2-\frac{1}{2}(N_2-1)(N_1+1)=\frac{1}{2}(N_1N_2+N_1-N_2+1).$$
\end{proof}
By Lemma \ref{lem4}, the computation of $\mathcal{M}(\bm{A}, \bm{B})$ is reduced to determine $N(0,1;t)\ (0\leq t\leq n_1+n_2-1)$. This is the key point in order to get the value of $\mathcal{M}(\bm{A}, \bm{B})$.

\begin{lemma}\label{lem5}
Assume $n_1, n_2\geq2$ and $\gcd(n_1, n_2)=1$. Then

(1). For $0\leq t\leq n_1+n_2-1$, $N(0,1;t)=U(t)-V(t)$, where
\begin{align*}
		U(t)= \frac{q_1q_2}{2^{t+3}}
\sum_{\substack{c'_0,c_0,c_1,\ldots,c_{t+1}\in\mathbb{F}_2
\\f_1(\alpha_1)=f_2(\alpha_2)=0}}(-1)^{c_0+c_{t+1}}, \quad
V(t)= \frac{q_2}{2^{t+2}}
\sum_{\substack{c_0,c_1,\ldots,c_{t+1}\in\mathbb{F}_2
\\f_2(\alpha_2)=0}}(-1)^{c_0+c_{t+1}},
	\end{align*}
and $f_1(x)=c'_0+c_1x+\cdots+c_{t+1}x^{t+1}$, $f_2(x)=c_0+c_1x+\cdots+c_{t+1}x^{t+1}$.

(2). \begin{align*}
V(t)=\left\{\begin{array}{ll}
2^{n_2-t-2},&\ \text{if}\ 0\leq t\leq n_2-2,\\
1,&\ \text{if}\ t=n_2-1,\\
0,&\ \text{if}\ t\geq n_2.
\end{array}\right.
\end{align*}

(3). $U(t)=U_1(t)+U_2(t)$, where
\begin{align*}
U_1(t)=\left\{\begin{array}{ll}
2^{n_1+n_2-t-3},&\ \text{if}\ 0\leq t\leq n_1+n_2-2,\\
\frac{1}{2},&\ \text{if}\ t=n_1+n_2-1,
\end{array}\right.\
U_2(t)=\left\{\begin{array}{ll}
0,&\ \text{if}\ 0\leq t\leq n_2+l-2\\
2^{n_1-l-2}(-1)^{f_0+1},&\ \text{if}\ t=n_2+l-1,\\
2^{n_1+n_2-t-3}(-1)^{f_0},&\ \text{if}\ n_2+l\leq t\leq n_1+n_2-2,\\
\frac{1}{2}(-1)^{f_0},&\ \text{if}\ t=n_1+n_2-1,
\end{array}\right.
\end{align*}
where $f_0$ and $l$ are determined by $F(x)=f_0+f_1x+\cdots+f_{l-1}x^{l-1}+x^l$ which is the unique polynomial in $\mathbb{F}_2[x]$
satisfying $F(x)g_2(x)\equiv 1\pmod{g_1(x)}$, $\deg F(x)=l\leq n_1-1$, and $g_1(x)$, $g_2(x)$ is the primitive polynomials in $\mathbb{F}_2[x]$ with degree $n_1$ and $n_2$ respectively and $g_1(\alpha_1)=g_2(\alpha_2)=0$.
\end{lemma}
\begin{proof}
(1).
From the definition of $N(0,1;t)$ we know
\begin{align*}
&N(0,1;t)\\
=&\#\{0\leq i\leq N_1N_2-1|T_1(\alpha_1^i)=0,T_2(\alpha_2^i)=1, T_1(\alpha_1^{i+\lambda})
=T_2(\alpha_2^{i+\lambda})\ (1\leq\lambda\leq t), T_1(\alpha_1^{i+t+1})+T_2(\alpha_2^{i+t+1})=1 \}\\
=&\frac{1}{2^{t+3}}\sum_{i=0}^{N_1N_2-1}(1+(-1)^{T_1(\alpha_1^i)})
(1-(-1)^{T_2(\alpha_2^i)})
\prod\limits_{\lambda=1}^t(1+(-1)^{T_1(\alpha_1^{i+\lambda})+T_2(\alpha_2^{i+\lambda})}
(1-(-1)^{T_1(\alpha_1^{i+t+1})+T_2(\alpha_2^{i+t+1})})\\
=&\frac{1}{2^{t+3}}\sum_{i=0}^{N_1N_2-1}\sum_{c_0'
,c_0,c_1,\ldots,c_{t+1}\in\mathbb{F}_2}(-1)^{c_0+c_{t+1}}\cdot
(-1)^{T_1(c_0'\alpha_1^i+c_1\alpha_1^{i+1}+\cdots+c_{t+1}\alpha_1^{i+t+1})
+T_2(c_0\alpha_2^i+c_1\alpha_2^{i+1}+\cdots+c_{t+1}\alpha_2^{i+t+1})}\\
=&\frac{1}{2^{t+3}}\sum_{c_0'
,c_0,c_1,\ldots,c_{t+1}\in\mathbb{F}_2}(-1)^{c_0+c_{t+1}}\sum_{i=0}^{N_1N_2-1}
(-1)^{T_1(\alpha_1^{i}f_1(\alpha_1))+T_2(\alpha_2^{i}f_2(\alpha_2))},
\end{align*}
where $f_1(x)=c'_0+c_1x+\cdots+c_{t+1}x^{t+1}$, $f_2(x)=c_0+c_1x+\cdots+c_{t+1}x^{t+1}\in\mathbb{F}_2[x]$.

Let $i=\lambda N_1+\mu N_2 (0\leq \lambda \leq N_2-1,\ 0\leq \mu\leq N_1-1)$, then
\begin{align*}
&N(0,1;t)\\
=&\frac{1}{2^{t+3}}\sum_{c_0'
,c_0,c_1,\ldots,c_{t+1}\in\mathbb{F}_2}(-1)^{c_0+c_{t+1}}
\left(\sum_{\lambda=0}^{N_2-1}(-1)^{T_2(\alpha_2^{\lambda N_1}f_2(\alpha_2))}\right)
\left(\sum_{\mu=0}^{N_1-1}
(-1)^{T_1(\alpha_1^{\mu N_2}f_1(\alpha_1))}\right)\\
=&\frac{1}{2^{t+3}}\sum_{c_0'
,c_0,c_1,\ldots,c_{t+1}\in\mathbb{F}_2}(-1)^{c_0+c_{t+1}}
\left(\sum_{x\in\mathbb{F}_{q_2}^*}(-1)^{T_2(xf_2(\alpha_2))}\right)
\left(\sum_{y\in\mathbb{F}_{q_1}^{*}}
(-1)^{T_1(yf_1(\alpha_1))}\right),
\end{align*}
since $f_1(\alpha_1)$ is independent of $c_0$, we get
\begin{align*}
&N(0,1;t)\\
=&\frac{1}{2^{t+3}}\sum_{c_0',c_0,c_1,\ldots,c_{t+1}\in\mathbb{F}_2}
(-1)^{c_0+c_{t+1}}
\left(\sum_{x\in\mathbb{F}_{q_2}}(-1)^{T_2(xf_2(\alpha_2))}-1\right)
\left(\sum_{y\in\mathbb{F}_{q_1}}
(-1)^{T_1(yf_1(\alpha_1))}-1\right)\\
=&\frac{1}{2^{t+3}}\left(\sum_{c_0',c_0,c_1,\ldots,c_{t+1}\in\mathbb{F}_2}
(-1)^{c_0+c_{t+1}}
\sum_{x\in\mathbb{F}_{q_2}}(-1)^{T_2(xf_2(\alpha_2))}
\left(\sum_{y\in\mathbb{F}_{q_1}}
(-1)^{T_1(yf_1(\alpha_1))}-1\right)\right.\\&\left.-\sum_{c_0',c_1,\ldots,c_{t+1}\in\mathbb{F}_2}
(-1)^{c_{t+1}}
\left(\sum_{y\in\mathbb{F}_{q_1}}
(-1)^{T_1(yf_1(\alpha_1))}-1\right)\sum_{c_0\in\mathbb{F}_2}(-1)^{c_0}\right)\\
=&\frac{1}{2^{t+3}}\sum_{c_0',c_0,c_1,\ldots,c_{t+1}\in\mathbb{F}_2}
(-1)^{c_0+c_{t+1}}
\sum_{x\in\mathbb{F}_{q_2}}(-1)^{T_2(xf_2(\alpha_2))}
\left(\sum_{y\in\mathbb{F}_{q_1}}
(-1)^{T_1(yf_1(\alpha_1))}-1\right)\\
=&\frac{q_2}{2^{t+3}}\sum_{\substack{c_0',c_0
,c_1,\ldots,c_{t+1}\in\mathbb{F}_2\\f_2(\alpha_2)=0}}(-1)^{c_0+c_{t+1}}
\left(\sum_{f_1(\alpha_1)=0}q_1-1\right)\\
=&\frac{q_1q_2}{2^{t+3}}\sum_{\substack{c_0',c_0
,c_1,\ldots,c_{t+1}\in\mathbb{F}_2\\f_1(\alpha_1)=f_2(\alpha_2)=0}}
(-1)^{c_0+c_{t+1}}-\frac{q_2}{2^{t+3}}\cdot2\cdot\sum_{\substack{c_0
,c_1,\ldots,c_{t+1}\in\mathbb{F}_2\\f_2(\alpha_2)=0}}(-1)^{c_0+c_{t+1}}\\
=&U(t)-V(t).
\end{align*}

(2). We now compute
\begin{align*}
V(t)= \frac{q_2}{2^{t+2}}
\sum_{\substack{c_0,c_1,\ldots,c_{t+1}\in\mathbb{F}_2
\\f_2(\alpha_2)=0}}(-1)^{c_0+c_{t+1}},
\end{align*}
where $f_2(x)=c_0+c_1x+\cdots+c_{t+1}x^{t+1}$.

From $f_2(\alpha_2)=0$ we get $f_2(x)=g_2(x)(\varepsilon_0+
\varepsilon_1x+\cdots+x^s)=\varepsilon_0+\cdots+x^{s+n_2}$.

If $t+1<n_2$, then $f_2(x)\equiv0$, $c_0=c_1=\cdots=c_{t+1}=0$. We get $V(t)=\frac{q_2}{2^{t+2}}=2^{n_2-t-2}$.
If $t=n_2-1$, then $f_2(x)\equiv0$ so that $c_0=c_{t+1}=0$ or $f_2(x)=g_2(x)$ so that $c_0=c_{t+1}=1$. We get $V(t)=\frac{q_2}{2^{t+2}}\cdot(1+1)=2^{n_2-t-1}=1$.
If $t>n_2-1$, for each $s$, $s+n_2\leq t+1$, $c_0=\varepsilon_0$ can be 0 or 1, we get $V(t)=M\cdot\sum_{c_0\in\mathbb{F}_2}(-1)^{c_0}=0$ where $M$ is independent of $c_0$.

(3). In the following we compute $U(t)=U_1(t)+U_2(t)$, where $U_1(t)$
 and $U_2(t)$ correspond to the cases of $c'_0=c_0$ and $c'_0\neq c_0$,
  respectively, that is
  \begin{align*}
U_1(t)=\frac{q_1q_2}{2^{t+3}}\sum_{\substack{c_0
,c_1,\ldots,c_{t+1}\in\mathbb{F}_2\\f_1(\alpha_1)=f_2(\alpha_2)=0}}
(-1)^{c_0+c_{t+1}}.
\end{align*}
From $c'_0=c_0$ we have $f_1(x)=f_2(x)=c_0+c_1x+\cdots+c_{t+1}x^{t+1}$, and $f_1(\alpha_1)=f_2(\alpha_2)=0$ is equivalent to
$g_1(x)g_2(x)\mid f_2(x)$ where $\deg g_1(x)g_2(x)=n_1+n_2$ and $\deg f_2(x)\leq t+1$.

If $t+1<n_1+n_2$, then $f_2(x)\equiv0$ which means $c_0=c_1=\cdots=c_{t+1}=0$. We get  $U_1(t)=\frac{q_1q_2}{2^{t+3}}=2^{n_1+n_2-t-3}$.
If $t+1=n_1+n_2$, then $f_2(x)\equiv0$, $c_0=c_{t+1}=0$ or $f_2(x)=g_1(x)g_2(x)$, $c_0=c_{t+1}=1$. We get $U_1(t)=\frac{q_1q_2}{2^{t+3}}\cdot(1+1)=2^{n_1+n_2-t-2}=\frac{1}{2}$.

For $c'_0\neq c_0$, we have $f_1(x)=f_2(x)+1$ and
  \begin{align*}
U_2(t)=\frac{q_1q_2}{2^{t+3}}\sum_{\substack{c_0
,c_1,\ldots,c_{t+1}\in\mathbb{F}_2\\g_2(x)\mid f_2(x)=c_0+c_1x+\cdots+c_{t+1}x^{t+1}\\
f_2(\alpha_1)=1}}(-1)^{c_0+c_{t+1}}.
\end{align*}

The conditions $g_2(x)|f_2(x)$ and $f_2(\alpha_1)=1$ are equivalent to $g_2(x)F(x)=f_2(x)\equiv 1\pmod{g_1(x)}$ for some $F(x)\in\mathbb{F}_2[x]$,
$0\leq l=\deg F(x)\leq \deg g_1(x)-1=n_1-1$, so that $F(x)=f_0+f_1x+\cdots+f_{l-1}x^{l-1}+x^l$. Such polynomial $F(x)$ in $\mathbb{F}_2[x]$ exists and is unique. Moreover, from $\deg f_2(x)\leq t+1\leq n_1+n_2=\deg g_1(x)g_2(x)$ we know that there are only two possibility for $f_2(x)$: $f_2(x)=g_2(x)F(x)$ and $f_2(x)=g_2(x)F(x)+g_1(x)g_2(x)$ and
for the latter case, we have $t+1=n_1+n_2$. Therefore,

$(\Rmnum{1}).$ For the case $t<n_1+n_2-1$, we get $f_2(x)=g_2(x)F(x)=f_0+\cdots+x^{n_2+l}$.

If $t<n_2+l-1$,  then $n_2+l=\deg f_2(x)\leq t+1\leq n_2+l-1$, which is impossible, we get $U_2(t)=0$.

If $t=n_2+l-1$, then $c_0=f_0$, $c_{t+1}=1$. We get
$U_2(t)=\frac{q_1q_2}{2^{t+3}}(-1)^{f_0+1}=2^{n_1-l-2}(-1)^{f_0+1}$.

If $n_2+l\leq t<n_1+n_2-1$, then $c_0=f_0$, $c_{t+1}=0$, and
$U_2(t)=\frac{q_1q_2}{2^{t+3}}(-1)^{f_0+0}=2^{n_1+n_2-t-3}(-1)^{f_0}$.

$(\Rmnum{2}).$ At last for the case $t=n_1+n_2-1$, then $n_2+l\leq t$.
When $f_2(x)=F(x)g_2(x)$ we have $c_0=f_0$, $c_{t+1}=0$. When $f_2(x)=g_2(x)F(x)+g_1(x)g_2(x)$,
we have $c_0=f_0+1$, $c_{t+1}=1$. We get $U_2(t)=\frac{q_1q_2}{2^{t+3}}\cdot2\cdot(-1)^{f_0+0}=\frac{1}{2}(-1)^{f_0}$. This completes the proof of Lemma \ref{lem5}.
\end{proof}

Now we come to prove Theorem 1.

\textbf{Proof of Theorem 1}

From Lemma \ref{lem4} we know that
$$\mathcal{M}({\bm{A}},{\bm{B}})=\frac{1}{2}(N_1N_2+N_1-N_2+1)-2\sum\limits_{t=0}^{n_1+n_2-1}tN(0,1;t),$$
where by Lemma \ref{lem5}, $N(0,1;t)=U_1(t)+U_2(t)-V(t)$ and
\begin{align*}
&\sum\limits_{t=0}^{n_1+n_2-1}t(U_1(t)+U_2(t))\\
=&2^{n_1+n_2}\left(\sum\limits_{t=1}^{n_1+n_2-2}\frac{t}{2^{t+3}}
+\frac{n_2+l-1}{2^{n_2+l+2}}(-1)^{f_0+1}+
\sum\limits_{t=n_2+l}^{n_1+n_2-2}\frac{t}{2^{t+3}}(-1)^{f_0}\right)
+\frac{n_1+n_2-1}{2}(1+(-1)^{f_0})\\
=&2^{n_1+n_2}\left[\frac{1}{8}(2-\frac{n_1+n_2}{2^{n_1+n_2-2}})
-\frac{n_2+l-1}{2^{n_2+l+2}}(-1)^{f_0}+\frac{1}{8}(-1)^{f_0}
(\frac{n_2+l-1}{2^{n_2+l-1}}-\frac{n_1+n_2}{2^{n_1+n_2-2}})\right]
+\frac{n_1+n_2-1}{2}(1+(-1)^{f_0})\\
&~~~~~~~~~~~~~~~~~~~~~~~~~~~~~~~~~~~~~~~~(\text{since\ $\sum\limits_{t=1}^{\lambda}\frac{t}{2^t}=2-
\frac{\lambda+2}{2^{\lambda}}\ (\lambda\geq1)$\ which can be proved by induction})\\
=&2^{n_1+n_2-2}+\frac{2^{n_l}}{2^{l+1}}(-1)^{f_0}-\frac{1}{2}(1+(-1)^{f_0}),
\end{align*}and
\begin{align*}
\sum\limits_{t=0}^{n_1+n_2-1}tV(t)
=2^{n_2}\sum\limits_{t=1}^{n_2-2}\frac{t}{2^{t+2}}+n_2-1
=2^{n_2-2}(2-\frac{n_2}{2^{n_2-2}})+n_2-1=2^{n_2-1}-1,
\end{align*}
then we obtain
\begin{align*}
\mathcal{M}({\bm{A}},{\bm{B}})&=\frac{1}{2}(2^{n_1+n_2}-2^{n_2+1}+2)
-2(2^{n_1+n_2-2}+2^{n_l-l-1}(-1)^{f_0}-\frac{1}{2}(1+(-1)^{f_0})-2^{n_2-1}+1)\\
&=(-1)^{f_0+1}(2^{n_1-l}-1).
\end{align*}
This completes the proof of Theorem 1.

\section*{Acknowledgment}
Keqin Feng's research was supported by the National Natural Science Foundation of China under Grant No:12031011.


\begin{thebibliography}{9}
	
	\bibitem{1C1} Z. Chen, Z. Niu, Y. Sang, and C. Wu, ``Arithmetic autocorrelation of binary $m$-sequences," \emph{Cryptologia}, Apr. 2022, doi: 10.1080/01611194.2022.2071116.
	
	\bibitem{2C2} Z. Chen, Z. Niu, and A. Winterhof, ``Arithmetic crosscorrelation of pseudorandom binary sequences of coprime periods," \emph{IEEE Trans. Inf. Theory}, vol. 68, no. 11, pp. 7538--7544, Nov. 2022.
	
	\bibitem{3G1} M. Goresky and A. Klapper, ``Arithmetic crosscorrelations of feedback with carry shift register sequences," \emph{IEEE Trans. Inf. Theory}, vol. 43, no. 4, pp. 1342--1345, Jul. 1997.
	
	\bibitem{4G2} M. Goresky and A. Klapper, ``Some results on the arithmetic correlation of sequences,"  in \emph{Sequences and Their Applications} (Lecture Notes
in Computer Science), vol. 5203. Berlin, Germany: Springer, 2018,
pp. 71--80.
	\bibitem{5G3} M. Goresky and A. Klapper, ``Statistical properties of the arithmetic correlation of sequences," \emph{Int. J. Found. Comput. Sci.}, vol. 22, no. 6, pp. 1297--1315, Sep. 2011.
	
	\bibitem{6G4} M. Goresky and A. Klapper, \emph{Algebraic Shift Register Sequences}. Cambridge, U.K.: Cambridge Univ. Press, 2012.
	

\bibitem{11J1} X. Jing, A. Zhang, and K. Feng, ``Arithmetic autocorrelation distribution of binary $m$-sequences," \emph{IEEE Trans. Inf. Theory}, vol. 69, no. 9, pp. 6040--6047, Sep. 2023.

	\bibitem{9Mand} D. Mandelbaum, ``Arithmetic codes with large distance," \emph{IEEE Trans. Inf. Theory}, vol. 13, no. 2, pp. 237--242, Apr. 1967.
	
\end{thebibliography}
\end{document}